\documentclass{llncs}

\usepackage{amssymb}
\usepackage{xspace}
\usepackage[english]{babel}
\usepackage{amsopn}
\usepackage{latexsym}
\usepackage{epsfig}
\usepackage{subfigure}
\usepackage{color}
\usepackage{enumerate}
\usepackage{tabularx}
\usepackage[plainpages=false]{hyperref}
\usepackage[figure,table]{hypcap}
\usepackage{multirow}
\usepackage{multicol}
\usepackage{graphicx}
\usepackage{booktabs} 
\usepackage{todonotes}
\usepackage{paralist} 
\usepackage{thm-restate}
\usepackage{wrapfig}
\usepackage[ruled,vlined,linesnumbered]{algorithm2e}
\usepackage[capitalise]{cleveref}

\graphicspath{{figures/}}

\crefname{property}{Property}{Properties} 

\newcommand{\algo}{\texttt{1PlanarTester}\xspace}
\renewcommand{\qed}{\hfill $\square$}

\title{
An Experimental Study of a $1$-planarity Testing and Embedding Algorithm
\thanks{Work partially supported by MIUR, under Grant 20174LF3T8 AHeAD: efficient
	Algorithms for HArnessing networked Data.
}
}

\author{
Carla Binucci \and
Walter Didimo \and
Fabrizio Montecchiani
}

\institute{
Universit\`a degli Studi di Perugia, Italy\\
\email{\{carla.binucci,walter.didimo,fabrizio.montecchiani\}@unipg.it} 
}

\begin{document}
\maketitle
\begin{abstract}
The definition of $1$-planar graphs naturally extends graph planarity, namely a graph is $1$-planar if it can be drawn in the plane with at most one crossing per edge. Unfortunately, while testing graph planarity is solvable in linear time, deciding whether a graph is $1$-planar is NP-complete, even for restricted classes of graphs. 
Although several polynomial-time algorithms have been described for recognizing specific subfamilies of $1$-planar graphs, no implementations of general algorithms are available to date. 
We investigate the feasibility of a $1$-planarity testing and embedding algorithm based on a backtracking strategy. While the experiments show that our approach can be successfully applied to graphs with up to 30 vertices, they also suggest the need of more sophisticated techniques  to attack larger graphs. Our contribution provides initial indications that may stimulate further research on the design of practical approaches for the $1$-planarity testing problem. 
\end{abstract}

\section{Introduction}\label{se:introduction}

The study of sparse nonplanar graphs is receiving increasing attention in the last years. One objective of this research stream is to extend the rich set of results about planar graphs to wider families of graphs that can better model real-world problems (see, e.g.,~\cite{DBLP:journals/algorithmica/Eppstein00,DBLP:conf/gis/Eppstein017,DBLP:journals/jea/EppsteinLS13,DBLP:journals/algorithmica/GrigorievB07,DBLP:journals/combinatorica/Grohe03}). Another objective is to create readable visualizations of nonplanar networks arising in various application scenarios (see, e.g.,~\cite{DBLP:journals/tcs/DidimoEL11,DBLP:journals/vlc/HuangEH14}). Both these motivations are embraced by a recent research topic in graph drawing and topological graph theory, which generalizes the notion of graph planarity to \emph{beyond-planar graphs}, informally defined as those graphs that can be drawn in the plane such that some prescribed edge crossing patterns are forbidden (see~\cite{DBLP:journals/jgaa/BekosKM18,DBLP:journals/csur/DidimoLM19,DBLP:journals/dagstuhl-reports/Hong0KP16,DBLP:journals/shonan-reports/HongT216} for surveys and reports). 

One of the most studied families of sparse nonplanar graphs, whose definition naturally extends that of planar graphs, is the family of $1$-planar graphs; refer to~\cite{DBLP:journals/csr/KobourovLM17} for a survey. A graph is \emph{$1$-planar} if it can be drawn in the plane such that each edge is crossed at most once.  An $n$-vertex $1$-planar graph has $O(n)$ edges~\cite{DBLP:journals/combinatorica/PachT97}, $O(\sqrt{n})$ separators and hence $O(\sqrt{n})$ treewidth~\cite{DBLP:journals/siamdm/DujmovicEW17,DBLP:journals/algorithmica/GrigorievB07}, and $O(1)$ stack and queue number~\cite{DBLP:journals/corr/AlamBK15,DBLP:journals/algorithmica/BekosBKR17,DBLP:journals/corr/abs-1904-04791}. 

Despite these (and other) similarities with planar graphs, recognizing whether a graph is $1$-planar is an NP-complete problem, in contrast with the well-known efficient algorithms for testing planarity.  A first proof was by Grigoriev and Boadlander~\cite{DBLP:journals/algorithmica/GrigorievB07}, and was based on a  reduction from the 3-partition problem. Another independent proof was later given by Korzhik and Mohar~\cite{KM13}, who used a reduction from the 3-coloring problem for planar graphs. The recognition problem for $1$-planar graphs is NP-complete even for graphs with bounded bandwidth, pathwidth, or treewidth~\cite{DBLP:journals/jgaa/BannisterCE18}; for graphs obtained from planar graphs by adding a single edge~\cite{DBLP:journals/siamcomp/CabelloM13}; and for graphs that come with a fixed rotation system~\cite{DBLP:journals/jgaa/AuerBGR15}. The problem becomes fixed-parameter tractable when parameterized by vertex-cover number, cyclomatic number, or  tree-depth~\cite{DBLP:journals/jgaa/BannisterCE18}.  Polynomial-time testing algorithms have been designed only for subfamilies of $1$-planar graphs (see, e.g.,~\cite{DBLP:journals/algorithmica/AuerBBGHNR16,DBLP:journals/algorithmica/Brandenburg18,DBLP:journals/algorithmica/Brandenburg19,DBLP:journals/algorithmica/HongEKLSS15} and refer to~\cite{DBLP:journals/csur/DidimoLM19,DBLP:journals/csr/KobourovLM17} for additional references and results).

\begin{figure}[t]
	\centering
	\subfigure[]{\includegraphics[width=0.4\columnwidth]{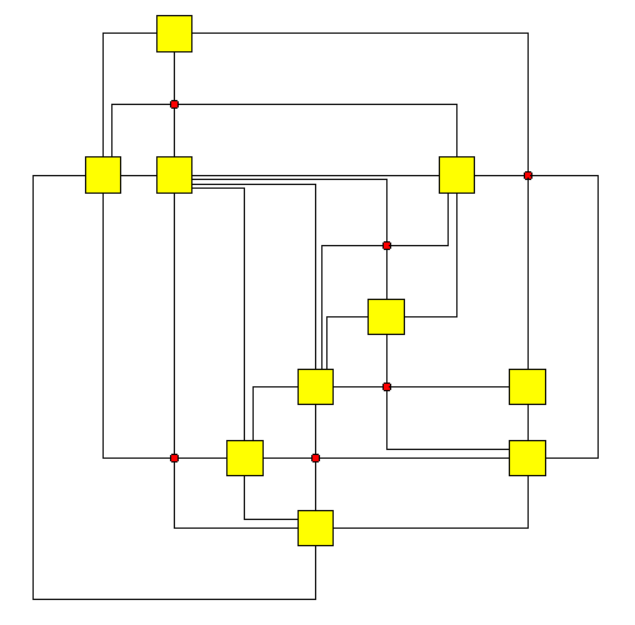}\label{fi:intro-a}}\hfil
	\subfigure[]{\includegraphics[width=0.4\columnwidth]{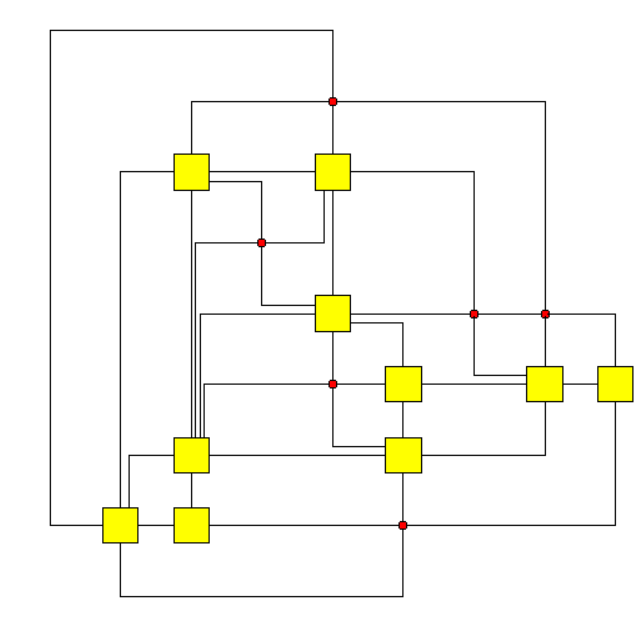}\label{fi:intro-b}}
	\caption{(a) A 1-planar drawing $\Gamma$ whose embedding is computed by our algorithm. (b) A drawing $\Gamma'$ of the same graph computed through a planarization approach. Both drawings contain six crossings (small circles), but $\Gamma'$ has two edges crossed twice.\label{fi:intro}}
\end{figure}

As witnessed by the above-mentioned literature, the problem of recognizing $1$-planar graphs has been studied by several authors aimed at drawing clear bounds between tractable and intractable instances. Nevertheless, there is a lack of general algorithms that can be effectively implemented and adopted in applications. The research presented in this paper goes in the direction of filling this gap by investigating practical approaches and by providing indications for further advances. Our contribution can be summarized as follows.


\smallskip\noindent $(i)$ We describe an easy to implement strategy based on a backtracking approach (\cref{se:algorithm}). It takes as input a general graph $G$ and decides whether $G$ is $1$-planar by exploring a search tree that encodes the space of candidate $1$-planar embeddings for the graphs.
If the test is positive, the algorithm returns a $1$-planar embedding of $G$ (see, e.g.,~\cref{fi:intro-a}). We remark that several algorithms designed for $1$-planar graphs assume that the input graph comes with a given $1$-planar embedding (see, e.g.,~\cite{DBLP:conf/gd/AlamBK13,DBLP:journals/tcs/BekosDLMM17,DBLP:journals/dcg/BiedlLM18,DBLP:journals/algorithmica/GiacomoDELMMW18,DBLP:conf/cocoon/HongELP12,DBLP:conf/gd/HongN16}), thus our algorithm can be used as a preliminary routine to compute such~an~embedding.

\smallskip\noindent $(ii)$ We report the results of an experimental study on a large set of instances with up to 50 vertices, taken from two well-established real-world graph benchmarks, the so-called \textsc{Rome} and \textsc{North} graphs~\cite{data,DBLP:journals/comgeo/WelzlBGLTTV97} (\cref{se:experiments}). The experiments indicate that while our approach can be successfully applied to the majority of instances with up to 30 vertices, more sophisticated techniques are needed to attack larger graphs. In terms of total number of edge crossings, our algorithm computes embeddings that on average contain no more than $1.8$ times the number of crossings of a state-of-the-art planarizer~\cite{DBLP:conf/gd/GutwengerM03}, which is allowed to cross an edge more than once; see, e.g.,~\cref{fi:intro}.

\smallskip\noindent $(iii)$ As a byproduct of our experiments, we make publicly available~\cite{url} the solved instances of the \textsc{Rome} and \textsc{North} graphs, with a labeling that specifies whether each instance is $1$-planar or not. An interesting finding is that the vast majority of the solved instances in these sets are $1$-planar, which corroborates the interest on $1$-planar graphs from an application perspective.

\section{Preliminaries}\label{se:preliminaries}

A \emph{drawing} $\Gamma$ of a graph $G$ maps each vertex of $G$ to a point of the plane and each edge of $G$ to a Jordan arc connecting its two endpoints. We only consider \emph{simple} drawings, where adjacent edges do not cross and where two independent edges cross at most in one of their interior points. A graph $G$ is \emph{planar} if it admits a \emph{planar} drawing, i.e., a crossing-free drawing. A planar drawing subdivides the plane into topologically connected regions, called \emph{faces}. The unbounded region is the \emph{outer face}. A \emph{planar embedding} of $G$ is an equivalence class of planar drawings of $G$ with a homeomorphic set of faces and the same face as outer face. A \emph{plane} graph is a graph with a given planar embedding. 
If $G$ is not planar, the \emph{planarization} of a drawing of $G$ is a plane graph obtained by replacing every crossing point with a \emph{dummy vertex}. An \emph{embedding} of $G$ is an equivalence class of drawings of $G$ whose planarizations yield the same planar embedding.

A graph is \emph{$1$-planar} if it admits a \emph{$1$-planar drawing}, that is, a drawing where each edge is crossed at most once.  A \emph{$1$-planar embedding} is the embedding induced by a $1$-planar drawing. A \emph{$1$-plane} graph is a graph with a given $1$-planar embedding. A \emph{kite} is a $1$-plane graph isomorphic to $K_4$, in which the outer face is bounded by a cycle composed of four vertices and four crossing-free edges, called \emph{kite edges}, while the remaining two edges cross (see, e.g.,~\cite{DBLP:journals/jgaa/Brandenburg14,DBLP:journals/algorithmica/GiacomoDELMMW18}). 
Similarly to planar graphs, a graph $G$ is $1$-planar if and only if all its subgraphs are $1$-planar; also, the following property holds.

\begin{property}\label{pr:2conn}
$G$ is $1$-planar if and only if every biconnected component of $G$ is $1$-planar. If $G$ is $1$-planar, a $1$-planar embedding of $G$ can be obtained in linear time by suitably merging the $1$-planar embeddings of its biconnected components. 
\end{property}

\section{Algorithm Design}\label{se:algorithm}

\newcommand{\cut}{\texttt{CUT}\xspace}
\newcommand{\sol}{\texttt{SOL}\xspace}
\newcommand{\cnt}{\texttt{CNT}\xspace}
\newcommand{\true}{\texttt{TRUE}\xspace}
\newcommand{\false}{\texttt{FALSE}\xspace}

We call our $1$-planarity testing and embedding algorithm \algo. In the following we first give an overview of this algorithm and then describe its backtracking procedure in detail (\cref{sse:backtracking}). Finally, we describe some further optimizations made to speed-up the algorithm. 

\algo takes as input a connected graph $G$ and works as follows. First, based on \cref{pr:2conn}, it computes the biconnected components of $G$ and processes each of them independently. For each biconnected component $C$, it executes some preliminary tests in order to verify whether $C$ can be immediately labeled as $1$-planar or as not $1$-planar. In the former case, \algo processes the next biconnected component, while in the latter case it halts and returns that $G$ is not $1$-planar. Namely, if $C$ is planar or it has less than $7$ vertices ($K_6$ is $1$-planar), then \algo labels $C$ as $1$-planar and computes a $1$-planar embedding of $C$. Otherwise, if $m_C > 4n_C-8$, where $n_C$ and $m_C$ are the number of vertices and edges of $C$, then \algo labels $C$ as not $1$-planar because it exceeds the maximum number of edges for a $1$-planar graph~\cite{DBLP:journals/combinatorica/PachT97}. If none of these conditions applies, \algo proceeds to the next step, which runs a backtracking procedure described in the following. The output of this step is a $1$-planar embedding of $C$, if it exists, or a negative answer. In the positive case, \algo processes the next biconnected component, otherwise it halts and returns that $G$ is not $1$-planar. At the end of this process, \algo will either output an embedding for each biconnected component of $G$, or it will return a component that is not $1$-planar.

\begin{figure}[t]
	\centering
	\subfigure[]{\includegraphics[width=0.32\columnwidth,page=1]{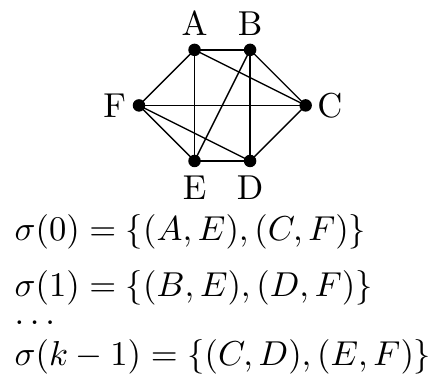}\label{fi:encoding}}
	\subfigure[]{\includegraphics[width=0.32\columnwidth,page=3]{backtracking}\label{fi:solution}}
	\subfigure[]{\includegraphics[width=0.33\columnwidth,page=2]{backtracking}\label{fi:backtracking}}
	\caption{(a) A graph $G$ and an ordering $\sigma$ of $\overline{E}$. (b) A planarization $G^*$ of $G$ (dummy vertices are gray squares) with the corresponding \true solution. (c) The search tree $T$.}
\end{figure}

\subsection{Backtracking Procedure}\label{sse:backtracking}

This procedure takes as input a biconnected graph $G=(V,E)$ with $n$ vertices and $m \leq 4n -8$ edges. 
We define as \emph{candidate solution} for the $1$-planarity testing problem on $G$ a set of pairs of crossing edges. Namely, let $\overline{E}$ be the set of all (unordered) pairs of edges that can cross in some embedding of $G$, i.e., the pairs $\{e_1,e_2\}$ such that $e_1 \in E$ and $e_2 \in E$ are independent (we only consider simple drawings). Let $k = |\overline{E}|$ and observe that $k  = O(n^2)$ because $m \le 4n-8$. Let $\sigma$ be any ordering of $\overline{E}$, and let $\sigma(i)$ denote the $i$-th pair of edges in such an ordering. We encode a candidate solution by a binary array $y$ of length $k$ such that $y[i]=0$ (resp. $y[i]=1$) means that the two edges $\sigma(i)$ do not cross (resp. cross) in a $1$-planar embedding of $G$ (if it exists). Refer to \cref{fi:encoding} for an illustration. We say that $y$ is a \true solution of $G$ if: (1) Each edge is crossed at most once, and (2) by replacing each crossing with a dummy vertex, the resulting graph $G^*$ is planar (see also \cref{fi:solution}). Otherwise we say that $y$ is a \false solution. Observe that Condition (2) is well defined because each edge is crossed at most once and hence we do not need to know the order of the crossings along the edges. 

\begin{lemma}\label{le:true-solutions}
	$G$ is $1$-planar if and only if the set of candidate solutions contains a \true solution. 
\end{lemma}
\begin{proof}
	If $G$ is $1$-planar, take a $1$-planar drawing $\Gamma$ of $G$. Consider a candidate solution $y$ such that $y[i]=0$ (resp. $y[i]=1$) if the two edges $\sigma(i)$ do not cross (resp. cross) in $\Gamma$. Since each edge of $\Gamma$ is crossed at most once, Condition~1 holds. Also, the planarization of $\Gamma$ induces a planar embedding of the graph $G^*$ obtained by replacing each crossing with a dummy vertex. Hence Condition~2 holds.  	
	
	Conversely, let $y$ be a \true solution. By Conditions~1 and~2, graph $G^*$ is planar and no two dummy vertices are adjacent. Let $\Gamma^*$ be a planar drawing of $G^*$. Consider a dummy vertex $d$ of $G^*$ that corresponds to a crossing of two edges $e_1=(u_1,v_1)$ and $e_2=(u_2,v_2)$ of $G$.
	In a clockwise visit of the neighbors of $d$, starting from $u_1$, the second vertex is either an end-vertex of $e_2$ or $v_1$. In the first case we replace $d$ with a crossing between $e_1$ and $e_2$ (\cref{fi:dummy_crossing1}). In the second case, $d$ does not really correspond to a crossing, thus we just remove it and redraw $e_1$ and $e_2$ without crossings, as in \cref{fi:dummy_crossing2}. Applying this transformation to every dummy vertex of $\Gamma^*$ we obtain a $1$-planar drawing~of~$G$.\qed  
\end{proof}

\begin{figure}[h]
	\centering
	\subfigure[]{\includegraphics[width=0.45\columnwidth,page=1]{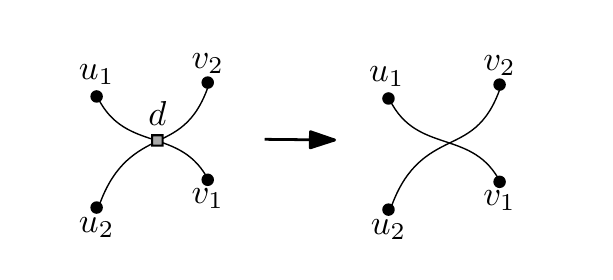}\label{fi:dummy_crossing1}}
	\subfigure[]{\includegraphics[width=0.45\columnwidth,page=2]{dummy_crossing}\label{fi:dummy_crossing2}}
	\caption{Replacement of a dummy vertex $d$: (a) $d$ is a crossing between $(u_1,v_1)$ and $(u_2,v_2)$; (b) $d$ is not a crossing between $(u_1,v_1)$ and $(u_2,v_2)$. }
\end{figure}

\begin{algorithm}
	\DontPrintSemicolon
	\SetNoFillComment
	\SetAlgoLined
	\SetKwInOut{Input}{Input}
	\SetKwInOut{Output}{Output}
	\SetKwProg{Algo}{VerifyNode($y_\nu$,$G$)}{}{}
	\Algo{}{
		\Input{An array $y_\nu$ encoding a partial candidate solution and a graph $G$.}
		\Output{One of \sol, \cut, \cnt.}
		\;
		\tcc{We start by checking the number of crossings per edge}	
		$cr\_pe$ $\leftarrow$ CountCrossingsPerEdge($y_\nu$,$G$)\;
		\uIf{$cr\_pe > 1$}{ 
			\Return \cut \tcp*{Too many crossings per edge}	
		}
		\;
		\tcc{We compute the sets of crossed edges and of kite edges}	
		$cross\_edges$ $\leftarrow$ FindCrossedEdges($y_\nu$,$G$)\;
		$kite\_edges$ $\leftarrow$ FindKiteEdges($y_\nu$,$G$)\;
		\uIf{$cross\_edges \cap kite\_edges \neq \emptyset$}{
			\Return \cut \tcp*{We don't want crossed kite edges}	
		}
		\;
		\tcc{We compute the graph induced by the saturated edges and then we replace its crossings with dummy vertices}	
		$G_\nu$ $\leftarrow$ ComputeInducedGraph($y_\nu$,$G$)\;	
		$G_\nu^*$ $\leftarrow$ CreateCrossVertices($y_\nu$,$G$)\;
		\eIf{isPlanar($G_\nu^*$)}{
			\eIf{$G_\nu \equiv G$}{ 
				\tcc{In this case the candidate solution is complete and a $1$-planar embedding is computed}
				\Return \sol 
			}
			{
				\tcc{We try to complete the partial candidate solution}
				$y_\nu^+$ $\leftarrow$ Complete($y_\nu$,$G$)\;
				$C_\nu^*$ $\leftarrow$ CreateCrossVertices($y^+_\nu$,$G$)\;
				\eIf{isPlanar($C_\nu^*$)}{
					\tcc{A $1$-planar embedding is computed}
					\Return \sol 	
				}
				{
					\Return \cnt\;
				}
			}
		}
		{
			\tcc{The array cannot be extended to a \true solution}	
			\Return \cut\ 
		}
		
	}
	\caption{Pseudocode of the VerifyNode routine.\label{a:check}}
\end{algorithm}

The backtracking procedure generates the set of candidate solutions incrementally, by computing a binary search tree $T$ as follows, see also \cref{fi:backtracking}. Each node $\nu$ of $T$ is equipped with an index $i_\nu < k$ and with an array $y_\nu$ of length $i_\nu$ that represents a partial candidate solution. Also, node $\nu$ has two children, $\nu_0$ and $\nu_1$, such that $i_{\nu_0}=i_{\nu_1}=i_\nu+1$, $y_{\nu_0}[i]=y_{\nu_1}[i]=y_{\nu}[i]$, for $i < i_\nu$, $y_{\nu_0}[i_\nu]=0$, and $y_{\nu_1}[i_\nu]=1$. We say that an array $y_\mu$ \emph{extends} $y_\nu$ if $y_\mu$ is associated with a node $\mu$ that is a descendant of $\nu$ in $T$. The search tree $T$ is traversed following a top-down order of the nodes starting from the root. When visiting a node $\nu$ of $T$, the backtracking procedure runs a routine called VerifyNode($y_\nu$,$G$), whose pseudocode is described by \cref{a:check} (some lines of this routine will be explained in the next subsection). This routine returns one of three possible values: \sol, \cut, \cnt. If the return value is \sol, $y_\nu$ is a \true solution (or can be extended to a \true solution) and the algorithm returns a $1$-planar embedding of $G$. If it is \cut, the algorithm will not visit the children of $\nu$ (i.e., the whole subtree of $T$ rooted at $\nu$ will be pruned) because $y_\nu$ is either a \false solution or it cannot be extended to a \true solution. If it is \cnt, we cannot conclude anything about $y_\nu$ and the procedure adds the two children $\nu_0$ and $\nu_1$ of $\nu$ to the set of nodes to be visited.

To describe in detail the VerifyNode($y_\nu$,$G$) routine, we need some definitions. An edge is \emph{crossed in $y_\nu$} if it is in a pair of edges $\sigma(j)$, with $j<i_\nu$, such that $y[j]=1$. An edge $e$ is \emph{saturated in $y_\nu$} if at least one of the following conditions applies: (a) $e$ is crossed in $y_\nu$; (b) the greatest index $j$ such that $\sigma(j)$ contains $e$ is smaller than $i_\nu$; (c) every pair of edges of $\overline{E}$ that contains $e$ is such that the other edge of the pair is crossed in $y_\nu$. 

\begin{lemma}\label{le:saturated}
Let $e$ be a saturated edge in $y_\nu$. Then $e$ is either crossed in $y_\nu$ or it is not crossed in every array $y_\mu$ that extends $y_\nu$.
\end{lemma}
\begin{proof}
If Condition~(a) holds the lemma follows, so assume that it does not hold. If Condition~(b) holds, $e$ will never appear in a $\sigma(j)$ with $j \geq i_\nu$, hence $e$ crosses in $y_\mu$ if and only if it crosses in $y_\nu$. Finally, if Condition~(c) holds then $e$ cannot cross any other edge, thus it is not crossed in any $y_\mu$ that extends $y_\nu$. \qed  
\end{proof}

First of all, the VerifyNode($y_\nu$,$G$) routine verifies that there is no edge crossed twice in $y_\nu$. Next, the routine computes the graph $G_\nu$ induced by the edges saturated  in $y_\nu$, and the graph $G_\nu^*$ obtained from $G_\nu$ by replacing each crossing with a dummy vertex. The following holds:

\begin{lemma}\label{le:induced}
Let $y_\mu$ be an array that extends $y_\nu$ (possibly coinciding with $y_\nu$) and that is a \true solution. Then $G_\nu^*$ is planar. 
\end{lemma}
\begin{proof}
If $y_\mu=y_\nu$, the statement follows, because $G_\mu$ corresponds to $G$ by construction, and thus $G_\nu^*=G_\mu^*$ is planar by definition of \true solution.
If $y_\mu \neq y_\nu$, $G_\nu^*$ is a subgraph of $G_\mu^*$ by \cref{le:saturated}. Namely, $G_\nu$ only contains the edges  saturated in $y_\nu$, thus such edges either already cross in $y_\nu$ or cross neither in $y_\nu$ nor in $y_\mu$. Since $G_\mu^*$ is planar, it follows that $G_\nu^*$ is also planar. \qed
\end{proof}

\noindent The routine verifies whether $G_\nu$ coincides with $G$. This is the case if $|y_\nu|=k$, but it may happen even if $|y_\nu|<k$  and all edges of $G$ are part of $G_\nu$ because are all saturated in $y_\nu$. If $G_\nu$ coincides with $G$, based on \cref{le:induced}, the routine tests if $G_\nu^*$ is planar. In the positive case, it returns \sol and a $1$-planar embedding of $G$ is obtained from a planar embedding of $G_\nu^*$ by replacing the dummy vertices with crossing points. In the negative case, the algorithm returns \cut because $y_\nu$ cannot be extended to a \true solution. If $G_\nu$ does not correspond to $G$, then the algorithm computes an arbitrary extension of $y_\nu$ and verifies whether it is a \true solution. This corresponds to exploring a leaf of the subtree of $T$ rooted at $\nu$. If  such a leaf is a \true solution, a $1$-planar embedding is computed and returned. Else \cnt is returned and the traversal of $T$ continues.

\medskip\noindent\textbf{Further Optimizations}. 
The main issue for the scalability of our algorithm is the quadratic length of the arrays representing the candidate solutions of our problem, which implies that the size of the search tree $T$ is $2^{O(n^2)}$. In order to reduce the search space, we first test whether there exists a small subset $S \subset E$ of edges of $G$, which we call \emph{skew edges}, whose removal yields a planar graph, i.e., we test whether $G$ is a $|S|$-skew graph (see, e.g.,~\cite{DBLP:journals/siamcomp/CabelloM13,DBLP:conf/compgeom/ChimaniH16}). The size of $S$ can be set as a parameter of our algorithm. If this is the case, \algo first executes the backtracking procedure by restricting the set $\overline{E}$ of pairs of edges that can cross to those pairs that contain at least one skew edge. This reduces the length of a candidate solution to $O(n \cdot |S|)$. If a solution exists, then $G$ has a $1$-planar embedding with at most $|S|$ crossings. Otherwise, we cannot conclude that $G$ is not $1$-planar because it may have a $1$-planar embedding where 2 edges of $E \setminus S$ cross each other, and hence we execute the standard backtracking procedure.

As a further optimization, we consider kite edges, defined as follows. An edge is a \emph{kite edge in $y_\nu$} if it is a kite edge with respect to a pair of edges that cross in $y_\nu$. In the VerifyNode routine, we extend the definition of a saturated edge $e$ by adding a fourth condition:  (d) $e$ is a kite edge in $y_\nu$. In fact, an array $y_\mu$ that extends $y_\nu$ and in which a kite edge is crossed can be discarded by the routine (see \cref{a:check}), because kite edges can always be redrawn without crossings. This allows us to increase the number of edges in $G_\nu$ and hence to increase the probability that $G_\nu$ coincides with $G$ or that $G^*_\nu$ is not planar. 

\section{Implementation and Experimental Analysis}\label{se:experiments}

\setlength{\tabcolsep}{2mm}
\renewcommand{\arraystretch}{1.1}

We implemented \algo in the C\# language. Our implementation exploits some planarity testing subroutines of the OGDF library~\cite{DBLP:reference/crc/ChimaniGJKKM13} (written in C++). We experimentally tuned the following parameters of our algorithm, which have an impact on the running time: (a) The strategy used to traverse the search tree $T$; (b) the size of the set of skew edges; (c) the probability that an arbitrary extension of a partial candidate solution is executed. Based on preliminary experiments, we chose the following settings. The search tree $T$ is traversed by a depth-first search, which guarantees that the space complexity of the algorithm is $O(n^2)$; note that with a breadth-first-search the space requirement may be $2^{\Omega(n^2)}$. We consider a set of skew edges of size one. The completion subroutine is executed with probability $0.8$. 


\smallskip\noindent\textbf{Experimental Goals.} 
The experimental analysis has two main objectives:

\begin{itemize}

\item[\textbf{$\bullet$ G1: Performance Analysis.}]  Evaluating the running time of our algorithm and the size of the largest instances it can handle in a reasonable time. Moreover, for those instances that are $1$-planar, we want to compare the total number of crossings produced by \algo with respect to a state-of-the-art planarizer that is allowed to cross an edge more than once. 

\item[\textbf{$\bullet$ G2: Labeling of Popular Graph Benchmarks.}] Labeling as $1$-planar or not $1$-planar the instances of well-established real-world graph benchmarks in graph drawing, namely the \textsc{Rome} and the \textsc{North} graphs~\cite{data,DBLP:journals/comgeo/WelzlBGLTTV97}. Our ultimate goal is to stimulate further practical research on beyond-planar graphs.

\end{itemize}

\smallskip\noindent\textbf{Experimental Setting.} We ran the experiments on a set of computers with uniform hardware and software features, namely computers with 16 GB of RAM, an Intel i7 CPU, and the Windows 10 operating system. We used the \textsc{Rome} and the \textsc{North} graphs as benchmarks for answering both \textbf{G1} and \textbf{G2}. We preliminary removed the planar instances from these two sets of graphs, as they are of no interest for us. \cref{ta:instances-labeling} reports some information about the considered nonplanar instances grouped by size. The first columns contain the number of instances in each sample, the average density (i.e., the average ratio between number of edges and number of vertices), and the average number of biconnected components, called blocks (recall that the algorithm processes each block independently).  We halted the computations that took more than 3 hours. As state-of-the-art planarizer, we used an implementation available in the OGDF library and discussed in~\cite{DBLP:conf/gd/GutwengerM03}, which makes use of heuristics described in~\cite{DBLP:journals/algorithmica/GutwengerMW05,DBLP:journals/tcad/JayakumarTS89}. To give an idea of the experimental effort, we report that, by summing up the running time of all computers, the experiments took more than 800 hours (i.e., more than one month). 
The output of the experiments in terms of labeling of the instances as $1$-planar or not $1$-planar is publicly available~\cite{url}.

\smallskip\noindent\textbf{Results and Discussion.} Concerning \textbf{G1}, \cref{ta:instances-labeling}, \cref{ta:runtime-stat}, and \cref{ta:crossings} summarize the performance of \algo in terms of number of solved instances, running time, and number of crossings, respectively. As it can be observed in the fifth and sixth columns of \cref{ta:instances-labeling}, our algorithm could solve a high percentage of instances up to 20 vertices (91.2\% of \textsc{Rome} and 73.6\% of  \textsc{North}). For the larger instances up to 40 vertices, the \textsc{North} graphs become more challenging, with a percentage of solved instances that is stable around 39\%. On the other hand, the percentage of solved instances on the \textsc{Rome} graphs is higher, namely above 69\% for the instances with 21-30 vertices, and above 43\% for the instances with 31-40 vertices.  We also ran our algorithm on instances with 41 to 50 vertices. \algo could solve 37.8\% of these instances in the \textsc{Rome} set and 18.8\% in the \textsc{North} set. 
The better performance of our algorithm on the \textsc{Rome} graphs is partly justified by their lower density with respect to the \textsc{North} graphs (see \cref{ta:instances-labeling}). Due to time constraints, we were able to process only a random subset of the non-planar instances with 31-50 vertices of the \textsc{Rome} set. 

\begin{table}[tbp]
	\centering
	\caption{\small{Instances and Labeling.}}\label{ta:instances-labeling}
	\scriptsize
	\resizebox{\textwidth}{!}{
		\begin{tabular}{l |r| r| r| r| r| r | r}
			
			\toprule
			& \textsc{\# } & \textsc{AVG } & \textsc{AVG  } & \multicolumn{2}{c|}{\textsc{Solved}} &  \multicolumn{2}{c}{\textsc{Label}}\\
			& \textsc{Instances} & \textsc{Density} & \textsc{\# Blocks}& Number & \% & \textsc{$1$-planar} & \textsc{Not $1$-planar}\\
			\midrule
			\textsc{Rome $10$-$20$} & $91$ & $1.49$ & $5.8$ & $83$ & $91.2\%$ & $100.0\%$ & $0.0\%$\\ 
			\textsc{Rome $21$-$30$} & $164$ & $1.37$ & $12.4$ & $114$ & $69.5\%$ & $100.0\%$ & $0.0\%$ \\ 
			\textsc{Rome $31$-$40$}  & $388$ & $1.31$ & $15.7$ & $170$ & $43.8\%$ & $100.0\%$ & $0.0\%$ \\ 
			\textsc{Rome $41$-$50$}  & $119$ & $1.29$ & $25.1$ & $45$ & $37.8\%$ & $100.0\%$ & $0.0\%$\\ 
			\textsc{North $10$-$20$} & $121$ & $2.05$ & $4.6$ & $89$ & $73.6\%$ & $88.8\%$  & $11.2\%$\\ 
			\textsc{North $21$-$30$} & $69$ & $2.07$ & $12.4$ & $27$ & $39.1\%$ & $77.8\%$ & $22.2\%$\\ 
			\textsc{North $31$-$40$} & $55$ & $2.00$ & $14.0$ & $21$ & $38.2\%$ & $57.1\%$ & $42.9\%$\\ 
			\textsc{North $41$-$50$} & $32$ & $1.79$ & $18.2$ & $6$ & $18.8\%$ & $83.3\%$ & $16.7\%$\\ 
			\bottomrule  
	\end{tabular} }
\end{table}

\begin{table}[tbp]
	\centering
	\caption{\small{Runtime and Backtracking stats.}}\label{ta:runtime-stat}
	\scriptsize
	\resizebox{\textwidth}{!}{
		\begin{tabular}{l| r| r| r| r| r| r| r| r| r}
			
			\toprule
			& \multicolumn{3}{c}{\textsc{Runtime (minutes)}} & \textsc{Solved by} & \multicolumn{2}{c|}{\textsc{Types of Solutions}} & \multicolumn{3}{c}{\textsc{Types of Cuts}} \\
			& \textsc{AVG} & \textsc{SD} & \textsc{max} & \textsc{Backtracking} & \textsc{Satur} & \textsc{Compl} & \textsc{DEC} & \textsc{KEC} & \textsc{Nonplanar}\\
			\midrule
			\textsc{Rome  $10$-$20$} &  $0.07$ & $0.60$ & $5.50$ & $14.5\%$ & $16.7\%$ & $83.3\%$ & $44.5\%$   & $30.5\%$	& $25.1\%$ \\ 
			\textsc{Rome  $21$-$30$}  & $0.38$ & $3.61$ & $38.40$ & $24.6\%$ & $25.0\%$ & $75.0\%$ & $51.7\%$   & $14.7\%$	& $33.6\%$ \\ 
			\textsc{Rome  $31$-$40$}  & $1.34$ & $12.06$ & $115.38$ & $21.2\%$ & $13.9\%$ & $86.1\%$ & $51.5\%$   & $14.5\%$	& $34.0\%$ \\ 
			\textsc{Rome  $41$-$50$}  & $0.01$ & $0.01$ & $0.02$ & $20.0\%$ & $22.2\%$ & $77.8\%$ & $45.9\%$   & $21.0\%$	& $33.1\%$ \\ 
			\textsc{North $10$-$20$}  & $1.86$ & $10.10$ & $92.69$ & $39.3\%$ & $11.4\%$ & $88.6\%$ & $43.1\%$	 & $40.5\%$ 	& $16.4\%$ \\ 
			\textsc{North $21$-$30$}  & $3.85$ & $10.14$ & $39.58$ & $25.9\%$ & $0.0\%$ & $100.0\%$  & $49.9\%$	 & $21.0\%$	& $29.2\%$ \\ 
			\textsc{North $31$-$40$}  & $3.78$ & $11.65$ & $39.80$ & $9.5\%$ & $0.0\%$ & $100.0\%$  & $42.0\%$	 & $45.9\%$	& $12.1\%$ \\
			\textsc{North $41$-$50$}  & $0.01$ & $0.01$ & $0.01$ & $0.0\%$  & $0.0\%$ & $0.0\%$  & $0.0\%$	 & $0.0\%$	    & $0.0\%$     \\ 
			\bottomrule  
	\end{tabular}}
\end{table}

About running time, \cref{ta:runtime-stat} reports the average, maximum, and standard deviation over all solved instances in each sample (we recall that the unsolved instances are those that took more than three hours). On average, a computation took less than 4 minutes, while the slowest instance took about 115 minutes. The standard deviation is about one order of magnitude greater than the mean value, but still around a few minutes, which indicates that the running time of \algo on the solved instances is relatively stable. It is worth noting that the running time on the largest graphs with up to 50 vertices is negligible with respect to the other values reported in the table and also very stable. By inspecting the raw data, we could observe that this behavior is due to the fact that such instances have many blocks (see also \cref{ta:instances-labeling}), most of which are solved through preliminary tests (i.e., without entering the backtracking routine) or are solved by backtracking but are very small. 



\cref{ta:runtime-stat} also reports some statistics about the computations, which may help to better understand the conditions that mostly influence the behavior of the backtracking algorithm. The fifth column shows the percentage of instances for which the algorithm needed to run the backtracking procedure, i.e., those for which the preliminary tests did not suffice. Many of these instances are in the subsets \textsc{Rome $21$-$30$}, \textsc{Rome $31$-$40$}, and \textsc{North $10$-$20$}.   
The columns \textsc{Types of Solutions} report the distribution of the conditions that led to a positive instance during an execution of the backtracking algorithm: \textsc{Satur} means that a solution was found by edge saturation, i.e., the subgraph $G_\nu$ induced by the current partial solution $y_\nu$ coincides with $G$ and the graph $G^*_\nu$ was planar; \textsc{Compl} means that a solution was found by running the completion subroutine on $y_\nu$ and again $G^*_\nu$ was planar (refer to \cref{a:check}). The data show that this second condition is much more frequent than the first one. We remark  that in the set \textsc{North $41$-$50$} we only solved those instances that did not require the backtracking algorithm, thus in this case we have 0\% both on \textsc{Satur} and on \textsc{Compl}. The columns \textsc{Types of Cuts} report the distribution of the cut conditions throughout a backtracking computation; they represent conditions that violate a valid solution: \textsc{DEC} represents a double edge crossing; \textsc{KET} represents a kite edge crossing; \textsc{Nonplanar} means that $G^*_\nu$ was nonplanar. We observe that the most frequent type of cuts is an edge crossed twice. Concerning the other two types, on the \textsc{Rome} set the \textsc{Nonplanar} condition is on average more frequent than the \textsc{KEC} condition, while we observe an inverse behavior on the \textsc{North} set. This seems to be partially related to the higher density of the \textsc{North} graphs, which intuitively yield a higher number of kite edges.

\begin{table}[tbp]
	\centering
	\caption{\small{Crossings.}}\label{ta:crossings}
	\scriptsize
	\resizebox{\textwidth}{!}{
		\begin{tabular}{l |r |r |r | r | r | r | r }
			\toprule
			& \multicolumn{6}{c|}{\textsc{Crossings per Graph}}  & \textsc{Crossing Ratio} \\
			&  \multicolumn{3}{c|}{\algo} & \multicolumn{3}{c|}{\textsc{OGDF}} &   \textsc{\algo~/~OGDF} \\
			&  \textsc{AVG} & \textsc{SD} & \textsc{max} &   \textsc{AVG} & \textsc{SD} & \textsc{max} &  \\%
			\midrule
			\textsc{Rome $10$-$20$} & $1.33$ & $1.36$ & $10$ & $1.13$ & $0.57$ & $5$ & $1.07$\\ 
			\textsc{Rome $21$-$30$} & $1.39$ & $1.58$ & $10$ & $1.16$ & $0.43$ & $3$ & $1.12$\\ 
			\textsc{Rome $31$-$40$} & $1.76$ & $2.76$ & $16$ & $1.13$ & $0.35$ & $3$ & $1.30$\\ 
			\textsc{Rome $41$-$50$} & $1.42$ & $2.16$ & $12$ & $1.24$ & $0.56$ & $3$ & $1.09$\\ 
			\textsc{North $10$-$20$} & $2.49$ & $1.99$ & $7$ & $1.78$ & $1.29$ & $7$ & $1.78$ \\ 
			\textsc{North $21$-$30$} & $2.90$ & $2.86$ & $9$ & $1.52$ & $0.85$ & $4$ & $1.64$ \\ 
			\textsc{North $31$-$40$} & $1.50$ & $1.12$ & $4$ & $1.50$ & $1.12$ & $4$ & $1.00$  \\ 
			\textsc{North $41$-$50$} & $1.00$ & $0.00$ & $1$ & $1.00$ & $0.00$ & $1$ & $1.00$  \\ 
			\bottomrule  
	\end{tabular} }
\end{table}

About the number of crossings, \cref{ta:crossings} reports the average, maximum, and standard deviation over all solved instances in each sample. Both \algo and OGDF produced drawings with very few crossings on average.  We observe that the \textsc{North} graphs yield more crossings than the \textsc{Rome} graphs, which still reflects the more challenging nature of the \textsc{North} set. Furthermore, the average over all ratios between the crossings made by \algo and those made by the OGDF planarizer is below 1.78 (last column in the table). This indicates that, on the solved instances, restricting the number of crossings per edge does not affect too much the total number of crossings. 

Concerning \textbf{G2}, \cref{ta:instances-labeling} reports the percentage of solved instances divided by $1$-planar and not $1$-planar. We recall that we removed all planar graphs from our benchmarks. All solved instances in the \textsc{Rome} set are $1$-planar, while the \textsc{North} graphs contain a percentage of non-$1$-planar instances varying between 11\% and 43\% over the different samples. The fact that in both sets the percentage of $1$-planar graphs is greater than the percentage of not $1$-planar graphs may also suggest that the unsolved instances are more likely to be not $1$-planar. To substantiate this hypothesis, we ran our algorithm on some graphs that are not $1$-planar but become $1$-planar after removing a few edges; these graphs have been generated by using minimal non-1-planar graphs described in~\cite{KM13}. The algorithm could not solve any of these instances within three hours, which confirms that our backtracking routine is not very effective on graphs having this kind of structure. 


\section{Conclusions and Future Research Directions}\label{se:conclusions}
We investigated the feasibility of a practical algorithm for 1-planarity testing and embedding. Our contribution provides initial indications that may stimulate further research on the design of more sophisticated approaches for the $1$-planarity testing problem.
Among the many directions that can be explored to improve the efficiency we suggest the following:

\smallskip\noindent $(i)$ A bottleneck of our algorithm is its $O(n^2)$ encoding scheme. One can try to reduce this size to $O(n)$, by using a balanced separating curves approach like that described in~\cite{DBLP:journals/jgaa/BannisterCE18}. This would allow us to treat larger instances, although it is unclear how to exploit this tool to design an encoding scheme that can be effectively integrated into a backtracking strategy.
	
\smallskip\noindent $(ii)$ More sophisticated rules can be studied to prune a subtree or to complete a partial solution during a backtracking execution. For example, can one use SPQR-trees to test different triconnected components independently and then merge them together?
	
\smallskip\noindent $(iii)$ One can study an ILP formulation for the problem. To this aim, the crossing minimization strategy in~\cite{DBLP:conf/esa/ChimaniMB08} may provide useful insights to compute crossing minimal 1-planar (or more in general $k$-planar) embeddings. 

A further research direction opened by our work is to extend our experiments to label a larger set of instances. For example, one could experimentally estimate the probability that a graph generated with a random or with a scale-free model is $1$-planar. 
Finally, the recognition problem is NP-complete for other families of beyond-planar graphs, such as fan-planar graphs (see~\cite{DBLP:journals/algorithmica/BekosCGHK17,DBLP:journals/tcs/BinucciGDMPST15}), and hence it would be interesting to design practical recognition algorithms also for these families.

\paragraph{Acknowledgments.} We thank Giuseppe Liotta for useful discussions.

{\small \bibliography{1planaritytest}}

\begin{thebibliography}{10}
\providecommand{\url}[1]{\texttt{#1}}
\providecommand{\urlprefix}{URL }
\providecommand{\doi}[1]{https://doi.org/#1}

\bibitem{data}
\url{http://www.graphdrawing.org/data.html}, [Online; accessed June-2019]

\bibitem{url}
\url{http://mozart.diei.unipg.it/montecchiani/1planarity/labels.xlsx}, [Online;
  accessed June-2019]

\bibitem{DBLP:conf/gd/AlamBK13}
Alam, M.J., Brandenburg, F.J., Kobourov, S.G.: Straight-line grid drawings of
  3-connected 1-planar graphs. In: Graph Drawing. LNCS, vol.~8242, pp. 83--94.
  Springer (2013)

\bibitem{DBLP:journals/corr/AlamBK15}
Alam, M.J., Brandenburg, F.J., Kobourov, S.G.: On the book thickness of
  1-planar graphs. CoRR  \textbf{abs/1510.05891} (2015)

\bibitem{DBLP:journals/algorithmica/AuerBBGHNR16}
Auer, C., Bachmaier, C., Brandenburg, F.J., Glei{\ss}ner, A., Hanauer, K.,
  Neuwirth, D., Reislhuber, J.: Outer 1-planar graphs. Algorithmica
  \textbf{74}(4),  1293--1320 (2016)

\bibitem{DBLP:journals/jgaa/AuerBGR15}
Auer, C., Brandenburg, F.J., Glei{\ss}ner, A., Reislhuber, J.: 1-planarity of
  graphs with a rotation system. J. Graph Algorithms Appl.  \textbf{19}(1),
  67--86 (2015)

\bibitem{DBLP:journals/jgaa/BannisterCE18}
Bannister, M.J., Cabello, S., Eppstein, D.: Parameterized complexity of
  1-planarity. J. Graph Algorithms Appl.  \textbf{22}(1),  23--49 (2018)

\bibitem{DBLP:journals/algorithmica/BekosBKR17}
Bekos, M.A., Bruckdorfer, T., Kaufmann, M., Raftopoulou, C.N.: The book
  thickness of 1-planar graphs is constant. Algorithmica  \textbf{79}(2),
  444--465 (2017)

\bibitem{DBLP:journals/algorithmica/BekosCGHK17}
Bekos, M.A., Cornelsen, S., Grilli, L., Hong, S., Kaufmann, M.: On the
  recognition of fan-planar and maximal outer-fan-planar graphs. Algorithmica
  \textbf{79}(2),  401--427 (2017)

\bibitem{DBLP:journals/tcs/BekosDLMM17}
Bekos, M.A., Didimo, W., Liotta, G., Mehrabi, S., Montecchiani, F.: On {RAC}
  drawings of 1-planar graphs. Theor. Comput. Sci.  \textbf{689},  48--57
  (2017)

\bibitem{DBLP:journals/jgaa/BekosKM18}
Bekos, M.A., Kaufmann, M., Montecchiani, F.: Guest editors' foreword and
  overview - special issue on graph drawing beyond planarity. J. Graph
  Algorithms Appl.  \textbf{22}(1),  1--10 (2018)

\bibitem{DBLP:journals/dcg/BiedlLM18}
Biedl, T.C., Liotta, G., Montecchiani, F.: Embedding-preserving rectangle
  visibility representations of nonplanar graphs. Discrete {\&} Computational
  Geometry  \textbf{60}(2),  345--380 (2018)

\bibitem{DBLP:journals/tcs/BinucciGDMPST15}
Binucci, C., {Di Giacomo}, E., Didimo, W., Montecchiani, F., Patrignani, M.,
  Symvonis, A., Tollis, I.G.: Fan-planarity: Properties and complexity. Theor.
  Comput. Sci.  \textbf{589},  76--86 (2015)

\bibitem{DBLP:journals/jgaa/Brandenburg14}
Brandenburg, F.J.: 1-visibility representations of 1-planar graphs. J. Graph
  Algorithms Appl.  \textbf{18}(3),  421--438 (2014)

\bibitem{DBLP:journals/algorithmica/Brandenburg18}
Brandenburg, F.J.: Recognizing optimal 1-planar graphs in linear time.
  Algorithmica  \textbf{80}(1),  1--28 (2018)

\bibitem{DBLP:journals/algorithmica/Brandenburg19}
Brandenburg, F.J.: Characterizing and recognizing 4-map graphs. Algorithmica
  \textbf{81}(5),  1818--1843 (2019)

\bibitem{DBLP:journals/siamcomp/CabelloM13}
Cabello, S., Mohar, B.: Adding one edge to planar graphs makes crossing number
  and 1-planarity hard. {SIAM} J. Comput.  \textbf{42}(5),  1803--1829 (2013)

\bibitem{DBLP:reference/crc/ChimaniGJKKM13}
Chimani, M., Gutwenger, C., J{\"{u}}nger, M., Klau, G.W., Klein, K., Mutzel,
  P.: The open graph drawing framework {(OGDF)}. In: Handbook of Graph Drawing
  and Visualization, pp. 543--569. Chapman and Hall/CRC (2013)

\bibitem{DBLP:conf/compgeom/ChimaniH16}
Chimani, M., Hlinen{\'{y}}, P.: Inserting multiple edges into a planar graph.
  In: {SoCG} 2016. LIPIcs, vol.~51, pp. 30:1--30:15. Schloss Dagstuhl (2016)

\bibitem{DBLP:conf/esa/ChimaniMB08}
Chimani, M., Mutzel, P., Bomze, I.M.: A new approach to exact crossing
  minimization. In: {ESA}. LNCS, vol.~5193, pp. 284--296. Springer (2008)

\bibitem{DBLP:journals/comgeo/WelzlBGLTTV97}
{Di Battista}, G., Garg, A., Liotta, G., Tamassia, R., Tassinari, E., Vargiu,
  F.: An experimental comparison of four graph drawing algorithms. Comput.
  Geom.  \textbf{7},  303--325 (1997)

\bibitem{DBLP:journals/algorithmica/GiacomoDELMMW18}
{Di Giacomo}, E., Didimo, W., Evans, W.S., Liotta, G., Meijer, H.,
  Montecchiani, F., Wismath, S.K.: Ortho-polygon visibility representations of
  embedded graphs. Algorithmica  \textbf{80}(8),  2345--2383 (2018)

\bibitem{DBLP:journals/tcs/DidimoEL11}
Didimo, W., Eades, P., Liotta, G.: Drawing graphs with right angle crossings.
  Theor. Comput. Sci.  \textbf{412}(39),  5156--5166 (2011)

\bibitem{DBLP:journals/csur/DidimoLM19}
Didimo, W., Liotta, G., Montecchiani, F.: A survey on graph drawing beyond
  planarity. {ACM} Comput. Surv.  \textbf{52}(1),  4:1--4:37 (2019)

\bibitem{DBLP:journals/siamdm/DujmovicEW17}
Dujmovic, V., Eppstein, D., Wood, D.R.: Structure of graphs with locally
  restricted crossings. {SIAM} J. Discrete Math.  \textbf{31}(2),  805--824
  (2017)

\bibitem{DBLP:journals/corr/abs-1904-04791}
Dujmovic, V., Joret, G., Micek, P., Morin, P., Ueckerdt, T., Wood, D.R.: Planar
  graphs have bounded queue-number. CoRR  \textbf{abs/1904.04791} (2019)

\bibitem{DBLP:journals/algorithmica/Eppstein00}
Eppstein, D.: Diameter and treewidth in minor-closed graph families.
  Algorithmica  \textbf{27}(3),  275--291 (2000)

\bibitem{DBLP:conf/gis/Eppstein017}
Eppstein, D., Gupta, S.: Crossing patterns in nonplanar road networks. In:
  {SIGSPATIAL/GIS}. pp. 40:1--40:9. {ACM} (2017)

\bibitem{DBLP:journals/jea/EppsteinLS13}
Eppstein, D., L{\"{o}}ffler, M., Strash, D.: Listing all maximal cliques in
  large sparse real-world graphs. {ACM} J. Exp. Algorithmics  \textbf{18}
  (2013)

\bibitem{DBLP:journals/algorithmica/GrigorievB07}
Grigoriev, A., Bodlaender, H.L.: Algorithms for graphs embeddable with few
  crossings per edge. Algorithmica  \textbf{49}(1),  1--11 (2007)

\bibitem{DBLP:journals/combinatorica/Grohe03}
Grohe, M.: Local tree-width, excluded minors, and approximation algorithms.
  Combinatorica  \textbf{23}(4),  613--632 (2003)

\bibitem{DBLP:conf/gd/GutwengerM03}
Gutwenger, C., Mutzel, P.: An experimental study of crossing minimization
  heuristics. In: Liotta, G. (ed.) {GD} 2003. LNCS, vol.~2912, pp. 13--24.
  Springer (2003)

\bibitem{DBLP:journals/algorithmica/GutwengerMW05}
Gutwenger, C., Mutzel, P., Weiskircher, R.: Inserting an edge into a planar
  graph. Algorithmica  \textbf{41}(4),  289--308 (2005)

\bibitem{DBLP:journals/algorithmica/HongEKLSS15}
Hong, S., Eades, P., Katoh, N., Liotta, G., Schweitzer, P., Suzuki, Y.: A
  linear-time algorithm for testing outer-1-planarity. Algorithmica
  \textbf{72}(4),  1033--1054 (2015)

\bibitem{DBLP:conf/cocoon/HongELP12}
Hong, S., Eades, P., Liotta, G., Poon, S.: F{\'{a}}ry's theorem for 1-planar
  graphs. In: {COCOON}. LNCS, vol.~7434, pp. 335--346. Springer (2012)

\bibitem{DBLP:journals/dagstuhl-reports/Hong0KP16}
Hong, S., Kaufmann, M., Kobourov, S.G., Pach, J.: Beyond-planar graphs:
  Algorithmics and combinatorics (dagstuhl seminar 16452). Dagstuhl Reports
  \textbf{6}(11),  35--62 (2016)

\bibitem{DBLP:conf/gd/HongN16}
Hong, S., Nagamochi, H.: Re-embedding a 1-plane graph into a straight-line
  drawing in linear time. In: Graph Drawing. LNCS, vol.~9801, pp. 321--334.
  Springer (2016)

\bibitem{DBLP:journals/shonan-reports/HongT216}
Hong, S., Tokuyama, T.: Algoritihmcs for beyond planar graphs. {NII} {Shonan}
  Meet. Rep.  (2016)

\bibitem{DBLP:journals/vlc/HuangEH14}
Huang, W., Eades, P., Hong, S.: Larger crossing angles make graphs easier to
  read. J. Vis. Lang. Comput.  \textbf{25}(4),  452--465 (2014)

\bibitem{DBLP:journals/tcad/JayakumarTS89}
Jayakumar, R., Thulasiraman, K., Swamy, M.N.S.: {O}$(n^2)$ algorithms for graph
  planarization. IEEE Trans. {CAD} Integr. Circuits Syst.  \textbf{8}(3),
  257--267 (1989)

\bibitem{DBLP:journals/csr/KobourovLM17}
Kobourov, S.G., Liotta, G., Montecchiani, F.: An annotated bibliography on
  1-planarity. Computer Science Review  \textbf{25},  49--67 (2017)

\bibitem{KM13}
Korzhik, V.P., Mohar, B.: Minimal obstructions for 1-immersions and hardness of
  1-planarity testing. J. Graph Theory  \textbf{72}(1),  30--71 (2013)

\bibitem{DBLP:journals/combinatorica/PachT97}
Pach, J., T{\'{o}}th, G.: Graphs drawn with few crossings per edge.
  Combinatorica  \textbf{17}(3),  427--439 (1997)

\end{thebibliography}
\bibliographystyle{splncs04}

\end{document}